%% file: recognition.tex
\documentclass{article}

\usepackage{tikz}
\usetikzlibrary{arrows,decorations.pathmorphing,backgrounds,positioning,fit,petri}

\tikzstyle{every node}=[circle, draw, fill=black,
inner sep=0pt, minimum width=3pt]

\usepackage{amsthm}

\usepackage{amsfonts}

\usepackage{amsmath}

\usepackage{algorithm}
\usepackage{algpseudocode}

\usepackage{hyperref}

\newtheorem{define}{Definition}[section]
\newtheorem{lemma}[define]{Lemma}
\newtheorem{corr}[define]{Corollary}

\hypersetup{
  linkcolor=green,        
  citecolor=green,        
  colorlinks=false,       
}

\begin{document}

\title{Recognition of unipolar and generalised split graphs}

\author{
  Colin McDiarmid\\
  \texttt{cmcd@stats.ox.ac.uk}
  \and
  Nikola Yolov\\
  \texttt{niklov@cs.ox.ac.uk}
}
\maketitle

\begin{abstract}
  A graph is unipolar if it can be partitioned into a clique and a disjoint union of cliques,
  and a graph is a generalised split graph if it or its complement is unipolar.
  A unipolar partition of a graph can be
  used to find efficiently the clique number, the stability number, the chromatic number,
  and to solve other problems that are hard for general graphs.
  We present the first $O(n^2)$ time algorithm
  for recognition of $n$-vertex unipolar and generalised split graphs,
  improving on previous $O(n^3)$ time algorithms.
\end{abstract}

\section{Introduction}
\subsection{Definition and motivation}

A graph is \emph{unipolar} if for some $k \ge 0$
its vertices admit a partition into $k+1$ cliques
$\{C_i\}_{i=0}^k$ so that there are no edges between $C_i$ and $C_j$ for $1 \le i < j \le k$.
A graph $G$ is a \emph{generalised split graph} if either $G$
or its complement $\overline{G}$ is unipolar.
All generalised split graphs are perfect;
and Pr{\"o}mel and Steger \cite{promelsteger}
show that almost all perfect graphs are generalised split graphs.
Perfect graphs can be recognised in polynomial time
\cite{perfect_recognition_1} and \cite{perfect_recognition_2},
and there are many NP-hard problems which are solvable in polynomial time
for perfect graphs,
including
the stable set problem,
the clique problem,
the colouring problem,
the clique covering problem
and their weighted versions \cite{algorithms_perfect}.
If the input graph is restricted to be a generalised split graph,
then there are much more efficient algorithms for the problems above \cite{history3}.
In this paper we address the problem of efficiently recognising generalised split graphs,
and finding a witnessing partition.

Previous recognition algorithms for unipolar graphs include
\cite{history1} which achieves $O(n^3)$ running time,
\cite{history2} with $O(n^2m)$ time
and \cite{history3} with $O(nm + nm^\prime)$ time,
where $n$ and $m$ are respectively the number of vertices and edges of the input graph,
and $m^\prime$ is the number of edges added after a triangulation of the input graph.
Note that almost all unipolar graphs and almost all generalised split graphs
have $(1 + o(1))n^2/4$ edges \cite{us}.
Further, by testing whether $G$ or $\overline{G}$ is unipolar,
each of the mentioned algorithms above recognises generalised split graphs 
in $O(n^3)$ time.
The algorithm in this paper has running time $O(n^2)$.

This leads to polynomial-time algorithms for the problems mentioned above
(stable set, clique, colouring and so on)
which have $O(n^{2.5})$ expected running time for a random perfect graph $R_n$
and an exponentially small probability of exceeding this time bound.
Here we assume that $R_n$ is sampled uniformly from the perfect graphs on vertex set
$[n] = \{1, 2, \ldots n\}$.

\subsection{Notation}

We use $V(G)$, $E(G)$, $v(G)$ and $e(G)$ to denote
$V$, $E$, $|V|$ and $|E|$ for a graph $G = (V, E)$.
We let $N(v)$ denote the neighbourhood of a vertex $v$, and
let $N^+(v)$ denote $N(v) \cup \{v\}$,
also called the closed neighbourhood of $v$.
If $G = (V, E)$ and $S \subseteq V$,
then $G[S]$ denotes the subgraph induced by $S$.
Let $\mathcal{GS}^+$ be the set of all unipolar graphs
and let $\mathcal{GS}$ be the set of all generalised split graphs.
Fix $G = (V, E) \in \mathcal{GS}^+$.
If $V_0, V_1 \subseteq V$ with $V_0 \cap V_1 = \emptyset$ and $V_0 \cup V_1 = V$
are such that $V_0$ is a clique and $V_1$ is a disjoint union of cliques,
then the ordered pair $(V_0, V_1)$ will be called a \emph{unipolar representation} of $G$
or just a \emph{representation} of $G$.
For each unipolar representation $R = (V_0, V_1)$
we call $V_0$ the \emph{central clique} of $R$,
and we call the maximal cliques of $V_1$ the \emph{side cliques} of $R$.
A graph is unipolar iff it has a unipolar representation.
\begin{define}
  Let $R = (V_0, V_1)$ be a unipolar representation of a graph $G$.
  A partition $\mathcal{B}$ of $V(G)$ is a block decomposition of $G$ with respect to $R$
  if the intersection of each part of $\mathcal{B}$ with $V_1$ is either
  a side clique or $\emptyset$.
\end{define}

\subsection{Plan of the paper}

Assume that $G$ is an input graph throughout.
The algorithm for recognising unipolar graphs has three stages.
In the first stage a sufficiently large maximal independent set is found.
The second stage constructs a partition $\mathcal{B}$ of $V(G)$,
such that $\mathcal{B}$ is a block decomposition
for some unipolar representation if $G \in \mathcal{GS}^+$.
The third stage generates a 2-CNF formula
which is satisfiable iff $\mathcal{B}$ is a block decomposition
for some unipolar representation.
The formula is constructed in such a way that a
satisfying assignment of the variables
corresponds to a representation of $G$,
and the algorithm returns either a representation of $G$,
or reports $G \notin \mathcal{GS}^+$.

We describe the third stage first (in \S \ref{sec:verify}) as it is
short and includes a natural transformation to 2-SAT.
In \S \ref{sec:indep} we discuss the first stage of finding large independent sets.
In \S \ref{sec:blocks}
we present the second stage when we seek to build a block decomposition.
Finally, in \S \ref{sec:random_perfect}, we briefly discuss random perfect graphs
and algorithms for them using the algorithm described above.

\subsection{Data Structures}

The most commonly used data type for this algorithm is the set.
We assume that the operation
$A \cap B$ takes $O(\min(|A|, |B|))$ time,
$A \cup B$ takes $O(|A| + |B|)$ time,
$A \setminus B$ takes $O(|A|)$ time
and $a \in A$ takes $O(1)$ time.
These properties can be achieved by using hashtables to implement sets.

Functions will always be of the form $f : [m] \rightarrow A$ for some $m$,
where $[m] = \{1, 2, \ldots m\}$.
Therefore functions can be implemented with simple arrays,
hence the lookup and assignment operations are assumed to require $O(1)$ time.

\section{Verification of Block Decomposition}
\label{sec:verify}
\subsection{2-SAT}

Let $x_1, \ldots, x_n$ be $n$ boolean variables.
A 2-\emph{clause} is an expression of the form $y_1 \lor y_2$,
where each $y_j$ is a variable,
$x_i$, or the negation of a variable, $\neg x_i$.
There are $4n^2$ possible 2-clauses.
The problem of deciding whether or not a formula of the form
$\psi = \exists x_1 \exists x_2 \ldots \exists x_n (c_1 \land c_2 \land \ldots \land c_m)$,
where each $c_j$ is a 2-clause, is satisfiable is called 2-SAT.
The problem 2-SAT is solvable in $O(n + m)$ time
-- \cite{even} and \cite{tarjan},
where $n$ is the number of variables and $m$ is the number of clauses
in the input formula.

\subsection{Transformation to 2-SAT}

Let $G = (V, E)$ be a graph with vertex set $V = [n]$.
In this subsection we show how to test if
a partition of $V$ is a block decomposition for some unipolar representation,
in which case we must have $G \in \mathcal{GS}^+$.
Let $\mathcal{B}$ be the partition of $V$ we want to test.
From each block of $\mathcal{B}$,
we seek to pick out some vertices to form the central clique $V_0$ of a representation,
with the remaining vertices in the blocks forming the side cliques.
Suppose that $|\mathcal{B}| = m$,
and $\mathcal{B}$ is represented by a surjective function $f : V \rightarrow [m]$,
so that $\mathcal{B} = \{f^{-1}[i] : i \in [m]\}$.
Let $\{x_v : v \in V\}$ be Boolean variables.
We use the procedure \texttt{verify} to construct a formula $\psi(x_1, \ldots x_n)$,
so that each satisfying assignment of $\{x_v\}$ corresponds to a representation of $G$.

\input{verify.alg}
There is an exception: the first time a clause is added to $\psi$
it should be added without the preceding sign for conjunction.
The following lemma is easy to check.
\begin{lemma}
  The formula $\psi$ is satisfiable
  iff $\mathcal{B}$ is a block decomposition for some representation.
  Indeed an assignment $\Phi: \{x_v : v \in V\} \rightarrow \{0, 1\}$ satisfies $\psi$
  if and only if $R = (V_0, V_1)$ is a representation of $G$
  and $\mathcal{B}$ is a block decomposition of $G$ with respect to $R$,
  where $V_i = \{v \in V: \Phi(v) = 1 - i\}$.
\end{lemma}
\begin{proof}
  Suppose $\Phi$ is a satisfying assignment and let $V_0$, $V_1$ be as above.
  If $u$ and $v$ are both in $V_0$, then $uv \in E$,
  since otherwise $\Phi$ contains a clause $\neg x_u \lor \neg x_v$.
  If $u$ and $v$ are in $V_1$,
  then either $uv \in E$ and $f(u) = f(v)$ or $uv \not\in E$ and $f(u) \neq f(v)$,
  because in the other two cases $\phi$ contains the clause $x_u \lor x_v$.
  This means that the vertices in $V_1$ are grouped into cliques by their value
  of $f$.
  For the other direction, it is sufficient to verify that each generated clause
  is satisfied, which is a routine check.
\end{proof}

At most a constant number of operations are performed per pair $\{u, v\}$,
so $O(n^2)$ time is spent preparing $\psi$.
The formula $\psi$ can have at most $2$ clauses per pair $\{u, v\}$,
so the length of $\psi$ is also $O(n^2)$, and since 2-SAT can be solved in linear time,
the total time for this step is $O(n^2)$.

\section{Independent Set}
\label{sec:indep}
\subsection{Maximum Independent Set of a Unipolar Graph}

Let $\alpha(G)$ be the maximum size of an independent set in a graph $G$.
Let $G \in \mathcal{GS}^+$ and let $R$ be a unipolar representation of $G$.
Observe that for any representation $R$ of $G$,
the number $s(R, G)$ of side cliques satisfies $s(R, G) \le \alpha(G) \le s(R, G) + 1$.
We deduce that for every two representations $R_1$ and $R_2$ of $G$ we have
$|s(R_1, G) - s(R_2, G)| \le 1$.

If for example $G$ is $K_n$ or its complement,
then the number $s(R, G)$ depends on $R$.
However, this is not necessarily the case for all graphs, see Figure 1.
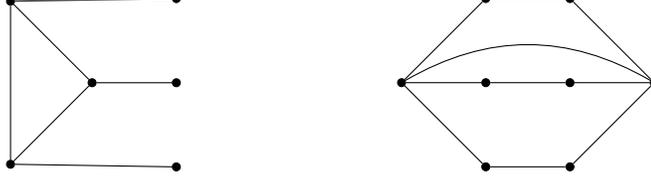
\begin{figure}[h]
  \caption{
    For the graph $G_1$ on the left, $s(R, G_1) = \alpha(G_1)$ for all representations $R$,
    and the graph $G_2$ on the right, $s(R, G_2) = \alpha(G_1) + 1$
    for all representations $R$.
  }
  \centering
  \hfill\hfill\hfill\hfill
  \begin{minipage}[b]{0.40\linewidth}
    \input{figure_a=s.tikz}
  \end{minipage}
  ~
  \begin{minipage}[b]{0.40\linewidth}
    \input{figure_a=s+1.tikz}
  \end{minipage}
\end{figure}
It can be shown that the number of $n$-vertex unipolar graphs with a unique representation
is $(1 - e^{-\Theta(n)}) | \mathcal{GS}^+_n|$,
and that the number of $n$-vertex unipolar graphs $G$ with a unique representation $R$
and such that $s(G, R) = \alpha(G)$
is $(1 - O( e^{-n^\delta } ) ) | \mathcal{GS}^+_n|$ for a constant $\delta > 0$
\cite{us}.

\subsection{Independent Set Algorithm}

It is well known that calculating $\alpha(G)$ for a general graph is NP-hard.
For $G \in \mathcal{GS}^+$ let $s(G) = \max_R s(R, G)$,
where the maximum is over all representations $R$ of $G$.
For $G \notin \mathcal{GS}^+$ set $s(G) = 0$.
In this section we see how to find a maximal independent set $I$,
such that if $G \in \mathcal{GS}^+$,
then $|I| \ge s(G)$ $(\ge \alpha(G) - 1)$.

The idea is to start with $G$ and with $I = \emptyset$;
and as long as the remaining graph  has two non-adjacent vertices, say $v_1$ and $v_2$,
pick $r = 1$ or $2$ of these vertices to add to $I$,
and delete from $G$ the closed neighbourhood of the added vertices.
We do this in such a way that a given representation $R$ of $G$ yields
a representation with $r$ less side cliques,
or (only when $r=2$) with one less side clique and the central clique removed.

\input{indep.alg}
Observe that the main body of the $\texttt{indep}(G)$ procedure is a \texttt{while} loop.
An alternative way of seeing the algorithm is that instead of the loop
there is a recursive call to $\texttt{indep}(G[U])$ at the end of the iteration
and the procedure returns the union of the
vertices found during this iteration and the recursively retrieved set.
A recursive interpretation is clearer to work with for inductive proofs.

\subsection{Correctness}

\begin{lemma}
  Procedure \label{corr_indep_maximal}
  $\textnormal{\texttt{indep}}(G)$ always returns a maximal independent set $I$.
\end{lemma}
\begin{proof}
  This is easy to see, since each vertex deleted from $U$ is adjacent to a vertex
  put in $I$.
\end{proof}

\begin{lemma}
  \label{corr_indep}
  If $\textnormal{\texttt{indep}}(G)$ returns $I$, then $|I| \ge s(G)$.
\end{lemma}

\begin{proof}
  If $G \notin \mathcal{GS}^+$, then the statement holds, because $s(G) = 0$.
  From now on assume that $G \in \mathcal{GS}^+$.
  We argue by induction on $v(G)$.
  It is trivial to see that the lemma holds for $v(G) = 1$.
  Let $v(G) > 1$ and assume that the lemma holds for smaller graphs.
  If $G$ is complete, then $|I| = 1 = s(G)$.

  Fix an arbitrary unipolar representation $R$ of $G$.
  We show that $|I| \ge s(R, G)$.
  If $G$ is not complete, then the procedure selects two non-adjacent vertices $u$ and $v$.
  The vertices $u$ and $v$ are either in different side cliques
  or one of them is in the central clique and the other is in a side clique.

  We start with the case when $u$ and $v$ are contained in side cliques.
  After inspecting their neighbourhoods,
  the algorithm removes from $U$ either one or both of them along with their neighbourhood.
  Suppose that it removes $r$ of them, where $r$ is $1$ or $2$.
  Let $G^\prime$ and $R^\prime$ be the graph and the representation
  induced by the remaining vertices.
  By the induction hypothesis, if $I^\prime$ is the recursively retrieved set,
  then $|I^\prime| \ge s(R^\prime, G^\prime)$.
  Let $I$ be the independent set returned at the end of the algorithm,
  so that $|I^\prime| + r = |I|$.
  Both $u$ and $v$ see all the vertices in their corresponding side clique,
  and see no vertices from different side cliques,
  so after removing $r$ of them with their neighbours,
  the number of side cliques in the representation decreases by precisely $r$,
  and hence $s(R, G) = s(R^\prime, G^\prime) + r$.
  Now $|I| = |I^\prime| + r \ge s(R^\prime, G^\prime) + r = s(R, G)$.

  Now w.l.o.g. assume that $u$ belongs to a side clique and $v$ belongs to the central
  clique.
  Then $N^+(u)$ contains the side clique of $u$ and perhaps parts of the central clique.
  Therefore, $N^+(u) \setminus N^+(v)$ is a subset of the side clique of $u$,
  and hence it is a clique.
  If $N^+(v) \setminus N^+(u)$ is a not clique, then the algorithm continues
  recursively with $G[V \setminus N^+(u)]$;
  and using the same arguments as above with $r = 1$, we guarantee correct behaviour.
  Now assume that $N^+(v) \setminus N^+(u)$ is a clique.
  Then $N^+(v) \setminus N^+(u)$ can intersect at most one side clique,
  because the vertices in different side cliques are not adjacent.
  In this case $N^+(v) \cup N^+(u)$ completely covers the side clique of $u$,
  completely covers the central clique,
  and it may intersect one additional side clique.
  Hence $s(R, G) = s(R^\prime, G^\prime) + 1$ or $s(R, G) = s(R^\prime, G^\prime) + 2$,
  where $G^\prime$ and $R^\prime$ are the induced graph and representation after
  the removal of $N^+(v) \cup N^+(u)$.
  If $I^\prime$ is the recursively obtained independent set,
  from the induction hypothesis we deduce that
  $|I| = |I|^\prime + 2 \ge s(R^\prime, G^\prime) + 2 \ge s(R, G)$.
\end{proof}

\subsection{Time Complexity}

In this form the algorithm takes more than $O(n^2)$ time, because checking
whether an induced subgraph is complete is slow.
However, we can maintain a set of vertices, $C$, which we have seen to induce a complete
graph.
We will create an efficient procedure to check if a subgraph is complete,
and to return some additional information to be used for future calls
if the subgraph is not complete.

\input{antiedge.alg}

The following lemma summarises the behaviour of $\texttt{antiedge}$.

\begin{lemma}
  \label{corr_antiedge}
  Let $C \subseteq U$ and suppose that $G[C]$ is complete.
  If $G[U]$ is complete, then $\textnormal{\texttt{antiedge}}(G, U, C)$ returns $(False, U)$;
  if not, then it returns
  $(uv, C^\prime)$ such that
  \begin{enumerate}
    \item
      $uv \in C^\prime \times (U \setminus C^\prime) - E(G)$
      i.e. $u \in C^\prime$, $v \in U \setminus C^\prime$, $uv \notin E(G)$
    \item
      $C \subseteq C^\prime \subseteq U$
    \item
      $G[C^\prime]$ is complete
  \end{enumerate}
\end{lemma}
\begin{proof}
  Easy checking.
\end{proof}

\input{mod_indep.alg}

\begin{lemma}
  \label{invariants_indep}
  Let $U$ and $C$ be the sets stored in the respective variables at the
  beginning of an iteration of the main loop of the modified $\textnormal{\texttt{indep}}$,
  and let $U^\prime$ and $C^\prime$
  be the sets stored at the beginning of the next iteration,
  if the algorithm does not terminate meanwhile.
  The following loop invariants hold:
  \begin{enumerate}
  \item
    [(I1)]
    $G[C]$ is complete and $C \subseteq U$,
  \item
    [(I2)]
    $U^\prime \setminus C^\prime \subseteq U \setminus C$.
  \end{enumerate}
\end{lemma}

A loop invariant is a condition which is true at the beginning of each iteration of a loop.

\begin{proof}
  Observe that the initial values of $U$ and $C$, which are $V$ and $\emptyset$ respectively,
  guarantee by Lemma \ref{corr_antiedge} that the values after the call to
  $\texttt{antiedge}$ satisfy condition (I1).
  Therefore (I1) holds for the first iteration.
  Concerning future iterations,
  observe that (I1) guarantees
  the precondition of Lemma \ref{corr_antiedge}, which it turn guarantees
  (I1) for the next iteration.
  We deduce that (I1) does indeed give a loop invariant.
  By proving this we have proved that the preconditions of Lemma \ref{corr_antiedge}
  are always met; and so we can use Lemma \ref{corr_antiedge} throughout.

  If $e = False$, then there is no next iteration,
  hence condition (I2) is automatically correct.
  Now assume that $e = u_1u_2$.
  Depending on $e_1$ and $e_2$ there are two cases for how many vertices are excluded.
  Case 1: one vertex is excluded.
  W.l.o.g. assume that $u_1$ is excluded,
  so $U^\prime = U \setminus N^+(u_1)$, $C \cap U_2 \subseteq C^\prime$
  and $C \subseteq N^+(u_1) \cup N^+(u_2)$.
  Then
  \begin{align*}
    U^\prime \setminus C^\prime
    &\subseteq U^\prime \setminus (C \cap U_2) \\
    &= (U \setminus N^+(u_1)) \setminus (C \cap (N^+(u_2) \setminus N^+(u_1)) \cap U)\\
    &= U \setminus [N^+(u_1) \cup (C \cap (N^+(u_2) \setminus N^+(u_1)))]\\
    &= U \setminus [N^+(u_1) \cup C] \subseteq U \setminus C.
  \end{align*}
  Case 2: two vertices are excluded.
  Now $U^\prime = U \setminus (N^+(u_1) \cup N^+(u_2))$, and
  $$
  U^\prime \setminus C^\prime \subseteq U^\prime \subseteq U \setminus C.
  $$
  We have shown that condition (I2) holds at the start of the next iteration,
  and so it gives a loop invariant as claimed.
\end{proof}

A vertex $v$ is \emph{absorbed} if it is processed during the
loop of $\texttt{antiedge}$ and then appended to the result set, $C^\prime$.

\begin{corr}
  \label{absorbed_vetrices}
  A vertex can be absorbed once at most.
  \begin{proof}
    Let $U, C, U^\prime$ and $C^\prime$ be as before.
    Observe that if, during the iteration, vertex $v$ is absorbed in a call to
    $\texttt{antiedge}$,
    then $v \in U \setminus C$ and $v \notin U^\prime \setminus C^\prime$.
    The Corollary now follows from the second invariant in Lemma \ref{invariants_indep}.
  \end{proof}
\end{corr}

\begin{lemma}
  \label{time_indep}
  The procedure
  $\textnormal{\texttt{indep}}(G)$ using
  $\textnormal{\texttt{antiedge}}$ takes $O(n^2)$ time.
\end{lemma}
\begin{proof}
  The total running time of each iteration of the main loop of $\texttt{indep}(G)$
  besides calling $\texttt{antiedge}$ is $O(n)$.
  The set $U$ decreases by at least one vertex on each iteration,
  so the time spent outside of $\texttt{antiedge}$ is $O(n^2)$.

  From Corollary \ref{absorbed_vetrices} at most $n$ vertices are absorbed
  and $O(n)$ steps are performed each time,
  so in total $O(n^2)$ time is spent in all calls to $\texttt{antiedge}$
  for absorbing vertices.

  Assume that $v$ is processed in $\texttt{antiedge}$ for the first time,
  but it is not absorbed and it is tested against a set $C_1$.
  Since $v$ is not absorbed,
  we may assume that $\texttt{antiedge}$ has returned the pair of vertices $vu$.
  At least one of $u$ and $v$ is removed (along with its neighbourhood) from $U$
  and moved to $I$.
  If $v$ is removed from $U$, then no more time can be spent on it in $\texttt{antiedge}$,
  hence the total time spent on $v$ in $\texttt{antiedge}$ is $O(n)$.
  Now assume that $u$ is removed.
  We have that $C_1 \subseteq N^+(u)$ and each vertex in $N^+(u)$ is removed from $U$.
  Hence, if $v$ is processed again in $\texttt{antiedge}$,
  it will be tested against a set $C_2$ with $C_1 \cap C_2 = \emptyset$,
  and therefore $|C_1| + |C_2| = |C_1 \cup C_2| = O(n)$.
  As we saw before, if $v$ is absorbed or removed from $U$,
  then it cannot be processed again in $\texttt{antiedge}$;
  and thus the running time spent on $v$ is again $O(n)$.
  If $v$ is not removed from $U$, then $C_2$ is removed from $U$.
  Hence, if $v$ is processed again in $\texttt{antiedge}$,
  $v$ will be tested against a set $C_3$ with
  $|C_1| + |C_2| + |C_3| = |C_1 \cup C_2 \cup C_3| = O(n)$, and so on.
  Thus, we see that over all these tests,
  each vertex is tested at most once for adjacency to $v$,
  and so the total time spent on $v$ is $O(n)$.
\end{proof}

\section{Building Blocks and Recognition}
\label{sec:blocks}
\subsection{Block Creation Algorithm}

In this subsection we present a short algorithm for creating a partition of $V(G)$ using
an independent set $I$ and then checking if this partition is a block decomposition
using the procedure \texttt{verify} from Section $2$.

\input{test.alg}
\begin{lemma}
  \label{corr_test}
  Suppose that $I \subseteq V(G)$ is an independent set with $|I| \ge s(G) - 1$ and
  $V_0 \cap I = \emptyset$ for some unipolar representation $R = (V_0, V_1)$ of $G$.
  Then $\textnormal{\texttt{test}}(G, I)$ returns $True$.
\end{lemma}
\begin{proof}
  On each step of the main loop a vertex from $i \in I$ is selected.
  Since $V_0 \cap I = \emptyset$, the vertex $i$ is a part of some side clique, say $C$.
  Now $C \cap N^+(j) = \emptyset$ for each $j \in I \setminus \{i\}$,
  so $C \subseteq U$.
  Also $C \subseteq N^+(i)$, and hence $C \subseteq N^+(i) \cap U$.
  Vertex $i$ does not see vertices from other side cliques, so $N^+(i) \cap U$
  is correctly marked as a separate block.

  Since $|I| \ge s(G) - 1$, at most one side clique is not represented in $I$.
  If there is an unrepresented side clique, say $C$,
  then none of the previously created blocks can claim any vertex from it,
  and hence $C \subseteq U$.
  We have shown that when the main loop ends, either $U \cap V_1 = \emptyset$
  or $U \cap V_1$ is a side clique;
  so $U$ is correctly marked as a separate block.
  The set $U$ also contains all remaining vertices, so $f$ is partition of $V$ into blocks,
  and hence $\texttt{verify}(G, f)$ will return $True$.
\end{proof}

\subsection{Block Decomposition Algorithm}

By Lemma and \ref{corr_indep_maximal} and \ref{corr_indep},
$\texttt{indep}(G)$ returns a maximal independent set $I$ of size at least $s(G)$.
Thus,
Lemma \ref{corr_test} suggests a naive algorithm for recognition for $\mathcal{GS}^+$ --
try $\texttt{test}(G, I \setminus i)$ for each $i \in I$ and return $True$
if any attempt succeeds.
The proposed algorithm is correct, since $|I \cap V_0| \le 1$.
The running time is $O(|I|n^2) = O(n^3)$, while we aim for $O(n^2)$.
However, with relatively little effort we can localise $I \cap V_0$ to at most
$2$ candidates from $I$.

\input{blocks.alg}
\paragraph{Procedure $\texttt{recognise}(G)$:}
\begin{algorithmic}
  \State \Return $\texttt{blocks}(G, \texttt{indep}(G))$
\end{algorithmic}

\subsection{Correctness}
\begin{lemma}
  \label{corr_blocks}
  The procedure $\textnormal{\texttt{recognise}}(G)$ returns $True$
  iff $G \in \mathcal{GS}^+$.
\end{lemma}
\begin{proof}
  First assume that $G \in \mathcal{GS}^+$ and
  let $R = (V_0, V_1)$ be an arbitrary representation of $G$.
  Let $I = \textnormal{\texttt{indep}}(G)$.
  By Lemma \ref{corr_indep},
  $|I| \ge s(G) \ge s(R, G)$.
  Since $V_0$ is a clique and $I$ is an independent set,
  we have $|V_0 \cap I| \le 1$.

  Case $1$: $V_0 \cap I = \emptyset$.
  Observe that $\texttt{blocks}$ returns $\texttt{test}(G, I^\prime)$,
  where $I^\prime$ is either $I$ or $I \setminus \{v\}$ for some $v \in I$,
  hence $|I^\prime| \ge |I| - 1 \ge s(G) - 1$;
  $I^\prime \subseteq I \subseteq V_1$,
  so $\texttt{test}(G, I^\prime) = True$ from Lemma \ref{corr_test}.

  Case $2$: $V_0 \cap I = \{c\}$.
  Blocks starts by calculating the set $C$,
  where $C = I$ if there is no $v \in V$ with $|N^+(v) \cap I| = 2$,
  and otherwise
  $$
  C = \bigcap \{N^+(v) \cap I : v \in V(G), |N^+(v) \cap I| = 2\}.
  $$

  Assume that $|N^+(v) \cap I| = 2$ for some $v \in V$.
  If $v \in V_0$, then $c \in N^+(v)$, because $V_0$ is a clique.
  If $v \in V_1$, then $N^+(v)$ can intersect at most one vertex from $I \cap V_1$
  and at most one vertex from $I \cap V_0 = \{c\}$ and since $|N^+(v) \cap I| = 2$,
  we have $c \in N^+(v)$.
  For each $v \in V$ if $|N^+(v) \cap I| = 2$, then $c \in N^+(v) \cap I$,
  so $c$ belongs to their intersection.
  If no $v \in V$ exists with $|N^+(v) \cap I| = 2$, then $C = I$, but $c \in I$,
  so again $c \in C$.
  We deduce that if $V_0 \cap I = \{c\}$, then $c \in C$ and $|C| > 0$.

  If $|C| = 1$ or $|C| = 2$ then $\texttt{test}(G, I \setminus \{i\})$ is
  tested individually for each vertex $i \in C$,
  but $c \in C$ and $\texttt{test}(G, I \setminus \{c\}) = True$ by Lemma \ref{corr_test}.

  If $|C| > 2$, then there is no $v \in V$ with $|N^+(v) \cap I| = 2$.
  Either $|I| = s(R, G)$ or $|I| = s(R, G) + 1$, so either all side cliques are represented
  by vertices of $I$, or at most one is not represented, say $S$.
  We can handle both cases simultaneously by saying that $S = \emptyset$ in the former case.
  We have that $I$ is a maximal independent set,
  but no vertex of $I \setminus \{c\}$ can see a vertex of
  $S$, because they belong to different side cliques,
  so $c$ is connected to all vertices of $S$ and therefore $\{c\} \cup S$ is a clique.
  Let $T = N(c) \cap (V_1 \setminus S)$.
  Then $|N^+(v) \cap I| = 2$ for each $v \in T$,
  but no such vertex exists by assumption,
  so $T = \emptyset$.
  Now $N(c) \cap V_1 = S$, and $V_1$ is a union of disjoint cliques, so
  $V_1 \cup \{c\}$ is also a union of disjoint cliques.
  Hence $R^\prime = (V_0 \setminus \{c\}, V_1 \cup \{c\})$ is a representation of $G$,
  so from Lemma \ref{corr_test} $\texttt{test}(G, I) = True$.

  On the contrary, if $G \notin \mathcal{GS}^+$, then there is no representation for $G$,
  hence $\texttt{test}$ cannot generate a block decomposition of $G$,
  and therefore $\texttt{test}$ will return $False$.
\end{proof}

\subsection{Time Complexity}
\begin{lemma}
  \label{time_recognise}
  $\textnormal{\texttt{recognise}}(G)$ takes $O(n^2)$ time.
\end{lemma}
\begin{proof}
  The procedure $\texttt{test}$ loops over a subset of $V$ and intersects two subsets of $V$,
  so the time for each step is bounded by $O(n)$, and since the number of steps is $O(n)$,
  $O(n^2)$ time is spent in the loop.
  Then it performs one more operation in $O(n)$ time, so the total time spent for preparation
  is $O(n^2)$.
  Then $\texttt{test}$ calls $\texttt{verify}$, which takes $O(n^2)$ time,
  so the total running time of $\texttt{test}$ is $O(n^2)$.

  While building $C$, $\texttt{blocks}$ handles $O(n)$ sets with size $O(n)$,
  so it spends $O(n^2)$ time in the first stage.
  Depending on the size of $C$, $\texttt{blocks}$ calls $\texttt{test}$ once or twice,
  but in both cases it takes $O(n^2)$ time,
  so the total running time of $\texttt{blocks}$ is $O(n^2)$.
  The total time spent for recognition is the time spent for $\texttt{blocks}$
  plus the time spent for $\text{indep}$, and since both are $O(n^2)$,
  the total running time for recognition is $O(n^2)$.
\end{proof}

\section{Algorithms for random perfect graphs}
\label{sec:random_perfect}

Gr{\"o}tschel, Lov{\'a}sz, and Schrijver \cite{algorithms_perfect}
show that
the stable set problem,
the clique problem,
the colouring problem,
the clique covering problem
and their weighted versions
are computable in polynomial time for perfect graphs.
The algorithms rely on the Lov{\'a}sz sandwich theorem,
which states that for every graph $G$ we have
$\omega(G) \le \vartheta(\overline{G}) \le \chi(G)$,
where $\vartheta(G)$ is the Lov{\'a}sz number.
The Lov{\'a}sz number can be approximated via the ellipsoid method in polynomial time,
and for perfect graphs we know that $\omega(G) = \chi(G)$,
hence $\vartheta(G)$ is an integer and its precise value can be found.
Therefore $\chi(G)$ and $\omega(G)$ can be found in polynomial time for perfect graphs,
though these are NP-hard problems for general graphs.
Further, $\alpha(G)$ and $\overline\chi(G)$ (the clique covering number)
can be computed from the complement of $G$ (which is perfect).
The weighted versions of these parameters
can be found in a similar way using the weighted version
of the Lov{\'a}sz number, $\vartheta_w(G)$.

These results tell us more about computational complexity
than algorithm design in practice.
On the other hand,
the problems above are much more easily solvable for generalised split graphs.
We know that the vast majority of the $n$-vertex perfect graphs are
generalised split graphs \cite{promelsteger}.
One can first test if the input perfect graph
is a generalised split graph using the algorithm in this paper
and if so, apply a more efficient solution.

Eschen and Wang \cite{history3} show that,
given a generalised split graph $G$ with $n$ vertices
together with a unipolar representation of $G$ or $\overline{G}$,
we can efficiently solve each of the following four problems:
find a maximum clique,
find a maximum independent set,
find a minimum colouring,
and find a minimum clique cover.

It is sufficient to show that this is the case when $G$ is unipolar,
as otherwise we can solve the complementary problem in the complement of $G$.
Finding a maximum size stable set and minimum clique cover in a unipolar graph
is equivalent to determining whether there exists a vertex in the central clique
such that no side clique is contained in its neighbourhood,
which is trivial and can be done very efficiently.
Suppose there are $k$ side cliques.
If there is such a vertex $v$, then a maximum size stable set
(of size $k + 1$) consists of $v$ and from each side clique a vertex not adjacent to $v$,
and a minimum size clique cover is formed by the central clique and the $k$ side cliques.
If not, then a maximum size stable set (of size $k$) consists of a vertex from each side
clique,
and a minimum clique cover is formed by extending the $k$ side cliques
to maximal cliques (which then cover $C_0$).

Let us focus on finding a maximum clique and minimum colouring of a unipolar graph $G$
with a representation $R$.
If $R$ contains $k$ side cliques, $C_1, \ldots C_k$,
then
$$
\omega(G)
= \chi(G)
= \max \{\omega(G[C_0 \cup C_i])\}_{i=1}^k
= \max \{\chi(G[C_0 \cup C_i])\}_{i=1}^k,
$$
where $C_0$ is the central clique.
Therefore, in order to find a maximum clique or a minimum colouring,
it is sufficient to solve the corresponding problem in each of the co-bipartite graphs
induced by the central clique and a side clique.
The vertices outside a clique in a co-bipartite graph
form a cover in the complementary bipartite graph,
and the vertices coloured with the same colour in a proper colouring of a co-bipartite graph
form a matching in the complementary bipartite graph.
By K\"onig's theorem it is easy to find a minimum cover using a given maximum matching,
and therefore finding a maximum clique and a minimum colouring in a co-bipartite graph
is equivalent to finding a maximum matching in the complementary bipartite graph.
For colourings, we explicitly find a minimum colouring in each co-bipartite graph
$G[C_0 \cup C_i]$,
and such colourings can be fitted together using no more colours,
since $C_0$ is a clique cutset.
Assume that $\overline{G[C_0 \cup C_i]}$ contains $n_i$ vertices and $m_i$ edges,
so each $n_i \leq n$ and $\sum_i m_i \leq |C_0| (n- |C_0|) \leq n^2 /4$.
We could use the Hopcroft--Karp algorithm for maximum matching in $O((|E| + |V|)\sqrt{|V|})$
time to find time bound $\sum_i O((m_i + n_i) \sqrt{n_i}) = O((n + m)\sqrt n) = O(n^{2.5})$.

The approach of Eschen and Wang \cite{history3} is very similar,
and they give more details,
but unfortunately there is a mistake with their analysis,
and a corrected version of their analysis yields $O( n^{3.5} / \log n)$ time,
instead of the claimed $O( n^{2.5} / \log n)$.
In order to see the mistake consider the case when the input graph is a split graph
with an equitable partition.

Given a random perfect graph $R_n$, we run our recognition algorithm in time $O(n^2)$.
If we have a generalised split graph, with a representation,
we solve each of our four optimisation problems in time $O(n^{2.5})$,
if not, which happens with probability $e^{-\Omega(n)}$,
we run the methods from \cite{algorithms_perfect}.
This simple idea yields a polynomial-time algorithm for each problem with
low expected running time,
and indeed the probability that the time bound is exceeded is exponentially small.

\bibliography{recognition}
\bibliographystyle{alpha}

\end{document}

%% file: figure_a=s.tikz
\begin{tikzpicture}
  \label{fig:big_stable_set}
  \node (a1)                      {};
  \node (a2)   [below right=of a1]{}
  edge                  (a1);
  \node (a3)   [below left =of a2]{}
  edge                  (a1)
  edge                  (a2);
  \node (b2)   [      right=of a2]{}
  edge                  (a2);
  \node (b1)   [above      =of b2]{}
  edge                  (a1);
  \node (b3)   [below      =of b2]{}
  edge                  (a3);
\end{tikzpicture}

%% file: figure_a=s+1.tikz
\begin{tikzpicture}
  \label{fig:big_stable_set}
  \node (a2)                      {};
  \node (b2)   [      right=of a2]{}
  edge                  (a2);
  \node (b1)   [above      =of b2]{}
  edge                  (a2);
  \node (b3)   [below      =of b2]{}
  edge                  (a2);
  \node (c1)   [      right=of b1]{}
  edge                  (b1);
  \node (c2)   [      right=of b2]{}
  edge                  (b2);
  \node (c3)   [      right=of b3]{}
  edge                  (b3);
  \node (d2)   [      right=of c2]{}
  edge                  (c1)
  edge                  (c2)
  edge                  (c3)
  edge     [bend right] (a2);
\end{tikzpicture}